\documentclass[11pt]{article}
\usepackage{fullpage}
\usepackage[T1]{fontenc}
\usepackage[utf8]{inputenc}
\usepackage{xcolor}
\usepackage[colorlinks=true]{hyperref}
\definecolor{red}{HTML}{F88379}
\hypersetup{
    linkcolor=blue
    ,citecolor=orange
    ,filecolor=purple
    ,urlcolor=orange
    ,menucolor=purple
    ,runcolor=purple
}
\usepackage{amsmath, amssymb, amsthm, amsfonts, graphicx, color, subcaption, enumerate, bm, array, mathtools}

\usepackage[capitalise,nameinlink]{cleveref}

\usepackage{multicol, multirow}
\usepackage{thmtools, thm-restate}
\usepackage{ragged2e}
\usepackage{lineno}
\usepackage{framed}
\usepackage[framemethod=tikz]{mdframed}
\usepackage[linesnumbered, ruled, vlined]{algorithm2e}
\usepackage{subfiles}
\graphicspath{{graphics/}}

\crefname{claim}{Claim}{Claims}
\crefname{property}{Property}{Properties}
\crefname{algocf}{Algorithm}{Algorithms}
\Crefname{algocf}{Algorithm}{Algorithms}

\usepackage[style=trad-alpha,natbib=true,maxcitenames=4]{biblatex}
\usepackage{dirtytalk}

\makeatletter
\g@addto@macro\bfseries{\boldmath}
\makeatother
\newtheorem{theorem}{Theorem}
\newtheorem{lemma}{Lemma}[section]
\newtheorem{remark}{Remark}
\newtheorem{definition}{Definition}
\newtheorem{corollary}{Corollary}
\newcommand{\cost}{\text{cost}}
\newcommand{\supp}{\text{support}}

\newcommand{\OPT}{\mathrm{OPT}}

\newcommand{\dens}{\mathrm{dens}}

\newcommand{\remove}[1]{}

\title{A Configuration-LP Framework for Connected $k$-Median Clustering}

\author{ Kushagra Chatterjee \footnote{Indian Statistical Institute, Kolkata, India. Email: kushagrachatterjee@gmail.com. This work was carried out when the author was a PhD student at the National University of Singapore } \and Rojin Rezvan \footnote {Virginia Tech, USA. Email: rojinrezvan@vt.edu } \and Ali Vakilian \footnote{Virginia Tech, USA. Email: vakilian@ttic.edu}}

\date{}

\begin{document}

\maketitle
\thispagestyle{empty}

\begin{abstract}
We study the \emph{connected $k$-median} clustering problem, a clustering problem that augments the classical $k$-median objective with connectivity constraints. We focus on the \emph{overlapping} variant of the problem, where clusters are allowed to share vertices. In addition to a metric space $(V,d)$, the input contains a connected graph $G$ on the same vertex set $V$ of size $n$. The goal is to select at most $k$ centers $C$ and assign vertices to them so as to minimize the $k$-median cost (i.e., $\sum_{v\in V} d(v,C)$), subject to the constraint that each cluster induces a connected subgraph of $G$. Since the metric space and the connectivity graph are independent, the problem is significantly more challenging than standard clustering. Eube et al.~\cite{eube2025esa} showed that even the assignment version is $\Omega(\log n)$-hard to approximate and gave approximation algorithms with guarantees depending polynomially on $k$.

We develop a configuration-LP-based framework that combines covering LP techniques with a rooted minimum-density oracle. For the assignment version, we obtain an $O(\log^2 n)$-approximation. For the general version, we develop a bicriteria framework that opens $O(k\log n)$ centers while achieving an $O(\log^2 n)$-approximation in cost. 
%Our results provide a different LP-based approach for handling connectivity constraints in clustering problems and demonstrate that configuration LPs, covering LPs, and rooted density oracles can be combined effectively to obtain approximation guarantees for clustering objectives under graph-theoretic constraints. 
\end{abstract}

\newpage
\bigskip
\tableofcontents
\bigskip
\thispagestyle{empty}

% \newpage

% \listoftodos

\newpage
\pagenumbering{arabic}
\section{Introduction}
Clustering is a central problem in algorithms and data analysis, with classical objectives such as $k$-center~\cite{hochbaum1985best, khuller2000capacitated} and $k$-median~\cite{charikar1999constant, charikar1999improved, chen2006k} extensively studied over the past decades. 
In the $k$-median problem, we are given a vertex set $V$ of size $n$ together with a metric $d: V \times V \rightarrow \mathbb{R}_{\ge 0}$ and an integer $k$, and the goal is to select $k$ centers so as to minimize the sum of distances of each vertex to its 
assigned center. While the vanilla $k$-median problem admits constant-factor approximation algorithms \cite{charikar1999constant, charikar1999improved, arya2001local, cohen2022improved}, many modern 
applications require incorporating additional structural constraints into the clustering process.

In this work, we study the \emph{connected $k$-median} problem, an extension of the classical $k$-median problem that incorporates \emph{connectivity constraints}. In addition to the metric space $(V,d)$, the input includes an unweighted undirected graph $G = (V,E)$, referred to as the \emph{connectivity graph}. The goal is to select $k$ centers and 
assign vertices to them so as to minimize the $k$-median objective, subject to the constraint that each cluster induces a connected subgraph in $G$. Importantly, the graph $G$ is completely independent of the metric $d$, making the problem fundamentally different from classical clustering.

We focus on the \emph{overlapping} variant of the connected $k$-median problem, where a vertex may belong to more than one cluster. This is natural in applications such as community detection \cite{bedi2016community, fani2017community}, where individuals can participate in multiple groups, and it also allows vertices such as articulation points to serve as connectors between several clusters. Theoretically, this distinction is crucial; Eube et al.~\cite{eube2025esa} showed that the disjoint version is extremely hard to approximate, admitting an $\Omega(n^{1-\varepsilon})$ lower bound even for $k=2$. 
%Thus, the overlapping model captures a meaningful and practically relevant relaxation while still retaining substantial algorithmic difficulty.

Connected clustering arises naturally in settings where both \emph{attribute data} and \emph{relational structure} are present. For instance, in applications such as community detection in social networks \cite{bedi2016community, fani2017community} or spatial partitioning tasks like school redistricting \cite{chen2023exploring}, the input consists of 
data points equipped with features as well as an underlying graph capturing relationships such as social interactions or geographic adjacency. In such settings, it is desirable to form clusters that are not only metrically coherent but also respect the connectivity structure of the graph. 
For example, in school redistricting, geographical units must be grouped into contiguous regions, ensuring that each cluster forms a connected area, while still optimizing criteria such as proximity or compactness.
This motivates the study of clustering models that integrate geometric objectives with graph-theoretic constraints.

From an algorithmic perspective, connectivity constraints significantly increase the difficulty of clustering problems. Even in the \emph{assignment version} of connected $k$-median, where the 
set of centers is fixed, and the goal is only to assign vertices to centers, Eube et al.~\cite{eube2025esa} showed that the problem is $\Omega(\log n)$-hard to approximate, and provided an $O(k \log n)$-approximation algorithm. For the general version, where centers must also be selected, they obtained an $O(k^2 \log n)$-approximation. These results demonstrate a substantial gap between $k$-median clustering and its connected variant.

\subparagraph{Our contributions.}
Our main contribution is a configuration-LP framework for connected $k$-median clustering. Unlike the cut-based relaxation of Eube et al.~\cite{eube2025esa}, where connectivity is enforced indirectly through constraints, our LP uses variables indexed by connected sets rooted at centers, thereby encoding connectivity directly into the variables themselves. We combine this formulation with covering-LP techniques and rooted minimum-density oracles to obtain approximation guarantees for connected $k$-median. Our main results are summarized below.

\begin{itemize}
    \item \textbf{Assignment version.} 
    We obtain an $O(\log^2 n)$-approximation algorithm for general graphs, which improves Eube et al \cite{eube2025esa} assignment version result when $k = \omega(\log n)$. %For minor-closed graph families, we obtain an $O(\log n)$-approximation, matching the known $\Omega(\log n)$ hardness bound. 

    \item \textbf{General version (bicriteria approximation).} 
    We develop a bicriteria framework that allows opening more than $k$ centers. In particular, we compute a solution that opens at most $O(k \log n)$ centers while achieving an $O(\log^2 n)$-approximation in cost, yielding an $(O(\log n),\, O(\log^2 n))$ bicriteria approximation. %For minor-closed graph families, this improves to an $(O(\log n),\, O(\log n))$ bicriteria approximation. 
\end{itemize}

The bicriteria guarantee highlights a tradeoff between the number of opened centers and the approximation factor in cost. The algorithm of Eube et al.~\cite{eube2025esa} gives a solution with $k$ centers with approximation factor of $O(k^2\log n)$ for the general version. In contrast, our bicriteria framework opens $O(k\log n)$ centers and achieves an $O(\log^2 n)$ cost approximation, eliminating the polynomial dependence on $k$ in the objective guarantee. %This suggests that the hard center-budget constraint is a major source of difficulty in the general connected $k$-median problem.

%Our results provide a unified LP-based framework for handling connectivity constraints in clustering problems. In particular, they demonstrate that configuration-LP and covering-LP techniques, combined with rooted density oracles, can be used to obtain approximation guarantees for clustering objectives under graph-theoretic constraints.

\subsection{Technical Overview}

Our algorithms are based on a configuration-LP viewpoint. Instead of assigning vertices to centers directly, we introduce a variable for each pair $(c,T)$, where $c$ is a center and $T$ is a connected vertex set containing $c$. Choosing such a configuration corresponds to opening a connected cluster rooted at $c$. This formulation captures the connectivity requirement directly: every integral solution selects connected sets, and the covering constraints ensure that every vertex is contained in at least one selected cluster. It is useful to contrast this formulation with the cut-based LP used by Eube et al.~\cite{eube2025esa}. In the cut-LP, the variables $x_v^c$ represent fractional assignments of vertices to centers, and connectivity is enforced indirectly through cut constraints: if $v$ is assigned to $c$, then every $v$--$c$ cut must contain enough mass assigned to $c$. This is a natural and compact relaxation, but it is still a \emph{vertex-assignment} formulation; after rounding, one must still ensure that the assigned vertices form connected clusters. In contrast, the configuration LP works directly with connected sets. A variable $y_{T,c}$ corresponds to selecting an entire connected cluster $T$ rooted at $c$, so connectivity is built into the variables themselves rather than enforced only through constraints.

This global representation is particularly useful for rounding. Once a configuration $(c,T)$ is sampled, it already defines a connected set containing its center, and the union of sampled configurations with the same root remains connected. Thus, the randomized rounding step only needs to ensure coverage of vertices; connectivity is preserved automatically. Moreover, every configuration-LP solution induces a feasible solution to the cut-LP by setting $x_v^c = \sum_{\substack{T\in \mathcal{T}_c\\ v\in T}} y_{T,c}$,
with the same objective value, so the configuration LP can be viewed as a more structured relaxation. The price of this stronger representation is that the LP contains exponentially many variables, so solving it exactly is not straightforwardly achievable in polynomial time. In fact, we can only solve this LP relaxation approximately, using column generation based on the rooted minimum-density connected-set oracle.

For the assignment version, where the center set $C$ is fixed, the configuration LP is a pure covering LP. We solve it approximately using the multiplicative weight update framework of Plotkin, Shmoys and Tarods~\cite{plotkin1991focs} for covering linear programs. The key step is the pricing problem: given weights $\alpha_v$ on the vertices that correspond to the covering constraints, we need to find a connected set $T \ni c$ minimizing the ratio between its cost and the total weight of vertices it covers. This is exactly a rooted minimum-density connected-set problem:
\[
\min_{T \mid T \ni c}
    \frac{\sum_{v\in T} d(v,c)}{\sum_{v\in T} \alpha_v}.
\]
Thus, a $\rho$-approximate oracle for rooted minimum-density connected sets yields a $\rho$-approximate pricing oracle for the configuration LP. Combining this with the Plotkin--Shmoys--Tardos framework gives a polynomial-support fractional solution whose cost is within a factor $O(\rho)$ of the optimum LP fractional value. We obtain such density oracles using connections to the {\em node-weighted Steiner tree} problem.

The rounding step is simple and uses the standard independent randomized rounding. Given a fractional solution $\hat y$, we sample each configuration $(c,T)$ with probability proportional to $\hat y_{T,c}\log n$. For each center $c$, the final cluster is the union of all the sampled configurations rooted at $c$. This union remains connected because every sampled set contains the common root $c$. The logarithmic oversampling ensures, by a union bound over
all vertices, that every vertex is covered with high probability. Linearity of expectation bounds the expected cost by an additional $O(\log n)$ factor. This yields an $O(\rho\log n)$-approximation for the assignment version. For general graphs, the rooted minimum-density problem admits an $O(\log n)$-approximation via known algorithms for the {\em budgeted node-weighted Steiner tree} problem. This yields an $O(\log^2 n)$-approximation for the assignment version. For minor-closed graphs, the techniques of~\cite{demaine2009node,chekuri2021node} suggest that the rooted minimum-density problem may admit an $O(1)$-approximation, which would improve the assignment-version guarantee to $O(\log n)$.

For the general version, we provide a bicriteria approximation guarantee. We avoid the center-rounding step of Eube et al \cite{eube2025esa} and instead formulate a non-fixed-center configuration LP. This LP has variables for all roots $c\in V$ and all
connected sets $T\in\mathcal T_c$, together with a fractional center-budget constraint
$
    \sum_{c\in V}\sum_{T\in\mathcal T_c} y_{T,c} \le k.
$
This LP is a mixed packing/covering program. Its pricing problem again reduces to a rooted minimum-density connected-set problem, except that the ratio includes both the connection cost of the configuration and the multiplier corresponding to the center-budget constraint. After computing a near-optimal fractional solution, we apply the same randomized rounding
scheme. The covering constraints guarantee that all vertices are covered with high probability, while the budget constraint implies that the number of sampled configurations, and hence the number of opened centers, is at most $O(k\log n)$ with high probability. This gives an $(O(\log n),O(\rho\log n))$ bicriteria approximation, which becomes $(O(\log n),O(\log^2 n))$ for general graphs.

\subsection{Related Work}

Clustering problems such as $k$-center and $k$-median are among the most well-studied optimization problems in algorithm design. Both problems are known to be NP-hard. For the $k$-center problem, a tight $2$-approximation algorithm is known~\cite{gonzalez1985clustering, hochbaum1985best}, and this factor is optimal unless P=NP~\cite{hsu1979easy}. 
For metric $k$-median, the current best polynomial-time approximation guarantee is $2+\varepsilon$, for any fixed $\varepsilon>0$, achieved by the recent algorithm of Cohen-Addad et al.~\cite{cohenSTOC25}. In contrast, no polynomial-time algorithm can achieve an approximation ratio smaller than $1+2/e\approx 1.736$, unless $\mathrm{P}=\mathrm{NP}$~\cite{jain2002new}.

\subparagraph*{Connected clustering.}
Clustering with connectivity constraints has been studied in several contexts. 
Ge et al.~\cite{ge2008joint} introduced the connected $k$-center problem as a model for jointly clustering attribute data and relationship data: the metric captures similarity between attribute vectors, while the connectivity graph captures relational information. They gave an exact dynamic program when the connectivity graph is a tree, and also studied general graphs, although without a known approximation guarantee. For general graphs with disjoint clusters, Drexler et al.~\cite{drexler2024connected} obtained an $O(\log^2 k)$-approximation; they also gave constant-factor approximations for Euclidean metrics of constant dimension and metrics of constant doubling dimension, motivated by applications in geodesy and sea-level analysis. In contrast, when overlapping clusters are allowed, connected $k$-center admits a simple $2$-approximation, highlighting an important distinction between the disjoint and overlapping models.

For the $k$-median objective, the connected variant has been studied more recently. Validi et al.~\cite{validi2022imposing} consider connected $k$-median in the context of districting, where geographical units must be grouped into contiguous regions while optimizing a proximity-based 
objective. They formulate the problem as an integer linear program, which can be solved optimally but requires exponential time in the worst case. 
Gupta et al.~\cite{gupta2011clustering} study the inapproximability of clustering with connectivity constraints and show that even when $k=2$ and the connectivity graph is star-like, the problem is $\Omega(\log n)$-hard to approximate. They also provide an $O(\log n)$-approximation algorithm for some restricted instances of the problem.

\subparagraph{Clustering with additional constraints.}
Beyond connectivity, a large body of work has focused on incorporating additional constraints into clustering problems motivated by real-world applications. Capacitated clustering introduces upper bounds on cluster sizes~\cite{li2016approximating}, while models with outliers allow for excluding a small fraction of points from the clustering~\cite{charikar2001algorithms}. More recently, fairness constraints have been extensively studied, requiring clusters to satisfy various fairness conditions~\cite{chierichetti2017fair,bera2019fair,backurs2019scalable,bercea2019on,dai2022fair,jung2020service,mahabadi2020individual,vakilian2022improved,ghadiri2021socially,abbasi2021fair,makarychev2021approximation,esmaeili2020probabilistic,esmaeili2021fair,chen2019proportionally}. Fairness constraints have also been studied for other types of clustering problems, like correlation clustering~\cite{ahmadian2020fair, ahmadi2020fair,ahmadian2023improved} and consensus clustering~\cite{chakraborty2025towards, chakraborty2026generalizing, chakraborty2026generic}. These variants highlight the increasing importance of structured clustering models, where combinatorial constraints are integrated with geometric objectives.

\section{Preliminaries and Problem Statement}
\label{sec:intro}

In this section, we introduce the definitions and notation used throughout the paper. We also formally state the problem and present our main results.

\subparagraph*{Problem Statement.}
We are given a vertex set $V$, a metric $d:V\times V\to\mathbb{R}_{\ge 0}$,
and a connected \emph{connectivity graph} $G=(V,E)$, which is unrelated to $d$. 
We consider the problem in two settings.

\begin{itemize}
    \item \textbf{Assignment Version:} A set of centers $C \subseteq V$ with $|C| \leq k$, for some integer $k$ is given. 
    A feasible \emph{overlapping} clustering solution outputs connected sets 
    $\mathcal{S}\coloneqq(S_c)_{c\in C}$ such that (i) $c \in S_c$ for all $c \in C$, 
    (ii) $G[S_c]$ is connected for all $c \in C$, and (iii) $\bigcup_{c\in C} S_c = V$.

    \item \textbf{General Version:} An integer $k$ is given, and the goal is to choose a set of 
    centers $C \subseteq V$ with $|C| \leq k$ together with connected sets 
    $\mathcal{S}\coloneqq(S_c)_{c\in C}$ satisfying the same conditions as above.
\end{itemize}

In both cases, the objective is to find a feasible solution $\mathcal{S}$ that minimizes
\[
\mathrm{cost}(\mathcal{S}) := \sum_{c\in C} \sum_{v\in S_c} d(v,c),
\]
where a vertex assigned to multiple clusters contributes once per cluster. 
Here $G[S_c]$ denotes the subgraph of $G$ induced by the vertex set $S_c$.

\begin{remark}\label{rem:connected}
Let $G_1,\dots,G_\ell$ be the connected components of $G$. In the assignment version, if some component contains at least one vertex but no center, then the instance is infeasible. Otherwise, each component can be solved independently, and the solutions can be combined. Hence, without loss of generality, we assume throughout that $G$ is connected and every vertex can reach some center.
\end{remark}

\begin{remark}[Bounded aspect ratio]
\label{ass:aspect-ratio}
Throughout the paper we assume the metric $d$ has polynomially bounded
aspect ratio: if
$
D_{\max} \ :=\ \max_{u,v\in V} d(u,v)$ and $\delta_{\min} \ :=\ \min_{u\neq v \in V} d(u,v),
$
we assume
\[
\Delta \coloneqq \frac{D_{\max}}{\delta_{\min}} = \mathrm{poly}(n).
\]
This is the standard bounded-spread condition used throughout the metric
clustering and embedding literature. All stated polynomial running times
in this paper are with respect to $n$, $k$, and $1/\varepsilon$ under this
assumption. Otherwise, the runtime of our algorithm polynomially depends on $\log(\Delta)$ too, which is polynomial in the input size.
\end{remark}

\subparagraph*{Our Results.}
We now formally state our results. A feasible overlapping clustering solution for either version is referred to as a \emph{connected $k$-median clustering} of the vertex set $V$. We say that a clustering $\mathcal{S}$ is an $\alpha$-approximate connected $k$-median clustering if $\cost(\mathcal{S}) \leq \alpha \cdot \cost(\mathcal{S}^*)$,
where $\mathcal{S}^*$ denotes an optimal connected $k$-median clustering.

For the assignment version, we obtain an improved approximation guarantee when $k = \omega(\log n)$. In particular, we show the following.

\begin{restatable}{theorem}{assignmentv}
\label{thm:assignment-version}
Given a vertex set $V$, a metric $d:V \times V \rightarrow \mathbb{R}_{\geq 0}$, an undirected, unweighted connectivity graph $G = (V, E)$, and a set of centers $C \subseteq V$ with $|C| \leq k$, there exists a polynomial-time algorithm that computes an $O(\log^2 n)$-approximate connected $k$-median clustering.
\end{restatable}

We now consider the general version of the problem, where the set of centers is not given and must be chosen as part of the solution. For the general version, we have obtained bicriteria approximation guarantees. For this, we need to define a bicriteria notion of approximation that allows the algorithm to open more than $k$ centers. A clustering $\mathcal{S}$ is said to be a $(\beta,\alpha)$-approximate connected $k$-median clustering if it uses at most $\beta k$ centers and satisfies $\cost(\mathcal{S}) \leq \alpha \cdot \cost(\mathcal{S}^*)$,
where $\mathcal{S}^*$ is an optimal solution that uses at most $k$ centers.

We obtain the following bicriteria guarantees.

\begin{restatable}{theorem}{bicriteria}
\label{thm:bicriteria}
Given a vertex set $V$, a metric $d:V \times V \rightarrow \mathbb{R}_{\geq 0}$, an undirected, unweighted connectivity graph $G = (V, E)$ and an integer $k$, there exists a polynomial-time algorithm that computes an $(O(\log n), O(\log^2 n))$-approximate connected $k$-median clustering.
\end{restatable}

\section{Assignment Version}

In this section, we prove our results for the assignment version of the problem. In particular, we prove \cref{thm:assignment-version}, which we restate below.

\assignmentv*

To prove \cref{thm:assignment-version}, we employ a \emph{configuration LP} formulation. Since this LP has exponentially many variables, we compute a near-optimal fractional solution using a \emph{rooted minimum-density connected set oracle} in conjunction with the Plotkin--Shmoys--Tardos~\cite{plotkin1991focs} multiplicative weight update framework for covering LPs. 
We then round this fractional solution to obtain the desired approximation guarantee. The details are presented in the subsequent sections.

\subsection{Configuration LP}
\label{sec:config}
We follow a \emph{configuration LP} route: formulate an exponential-size covering LP over all connected
clusters, solve it approximately via column generation/packing-covering with a \emph{minimum-density} oracle,
and round it in one shot via randomized rounding. This yields an $O(\rho\log n)$ approximation given a
$\rho$-approximation oracle for the rooted minimum-density connected set problem.

For each center $c\in C$, let $\mathcal{T}_c$ denote the family of all connected vertex sets
$T\subseteq V$ such that $c\in T$ and $G[T]$ is connected.
For $T\in\mathcal{T}_c$, define
\[
\mathrm{cost}_c(T) \coloneqq \sum_{v\in T} d(v,c).
\]

We introduce a variable $y_{T,c}\ge 0$ for each pair $(c,T)$ to represent the extent to which we choose configuration $(c,T)$. We define the primal and dual relaxation of the assignment problem as follows.

\paragraph*{Primal (Configuration) LP:}
\begin{align}
\min \quad & \sum_{c\in C} \sum_{T \in \mathcal{T}_c} \mathrm{cost}_c(T)\, y_{T,c}
\label{eq:clp-obj}\\
\text{s.t.}\quad
& \sum_{c\in C} \sum_{\substack{T \in \mathcal{T}_c \\ v \in T}} y_{T,c} \ge 1
&& \forall v\in \textcolor{blue}{V \setminus C},
\label{eq:clp-cover}\\
& y_{T,c} \ge 0
&& \forall c \in C,\ \forall T \in \mathcal{T}_c.
\label{eq:clp-box}
\end{align}

\paragraph*{Dual LP:}
\begin{align}
\max \quad & \sum_{v\in \textcolor{blue}{V \setminus C}} \alpha_v
\label{eq:dual-obj}\\
\text{s.t.}\quad
& \sum_{v \in T \textcolor{blue}{\setminus \{c\}}} \alpha_v \le \mathrm{cost}_c(T)
&& \forall c\in C,\ \forall T \in \mathcal{T}_c,
\label{eq:dual-pack}\\
& \alpha_v \ge 0
&& \forall v \in \textcolor{blue}{V \setminus C}.
\label{eq:dual-box}
\end{align}

\begin{lemma}
\label{lem:clp-relaxation}
The configuration LP \eqref{eq:clp-obj}--\eqref{eq:clp-box} is a valid relaxation of the assignment version of connected $k$-median with overlapping clusters. In particular, its optimal value is at most $\OPT$.
\end{lemma}

\begin{proof}
    We provide the proof in the appendix.
\end{proof}
%==========================================================
\subsection{The Oracle: Rooted Minimum-Density Connected Set}
\label{sec:oracle}
Since the configuration LP has exponentially many variables (one for each connected set $T\in\mathcal{T}_c$ and each center $c\in C$), it is not immediate how to solve it in polynomial time. We follow the standard oracle-based approach for large-scale packing/covering linear programs: rather than enumerating all configurations explicitly, we access them through an optimization subroutine that, for suitable dual weights, returns an approximately best column~\cite{plotkin1991focs,young2001sequential}.
In our setting, this subroutine is the {\em rooted minimum-density connected-set problem} defined below.

\begin{definition}%[Rooted minimum-density connected set]
\label{def:rooted-density}
Given a graph $G=(V,E)$, a root $c\in V$, nonnegative node \emph{profits}
$\alpha:V\to\mathbb{R}_{\ge 0}$, and nonnegative node \emph{costs}
$w:V\to\mathbb{R}_{\ge 0}$, the rooted minimum-density problem asks for a connected set
$T\ni c$ minimizing
\[
\dens(T) \coloneqq \frac{\sum_{v\in T} w(v)}{\sum_{v\in T}\alpha_v},
\]
with the convention $\dens(T)=+\infty$ if  $\sum_{v\in T}\alpha_v=0$.
\end{definition}
In our application, the dual variables are defined only on $V\setminus C$, so we set $\alpha_c:=0$.
Thus $\sum_{v\in T}\alpha_v=\sum_{v\in T\setminus\{c\}}\alpha_v$, which matches the left-hand side of the
dual packing constraint \eqref{eq:dual-pack}. In particular, minimizing $\dens(T)$ identifies the
most violated dual constraint, and hence the most promising column for the primal configuration LP.

\begin{definition}%[$\rho$-approximate density oracle]
\label{def:rho-oracle}
A $\rho$-approximate density oracle returns a connected set $T\ni c$ such that
$
\dens(T)\ \le\ \rho\cdot \min_{T'\ni c}\dens(T').
$
\end{definition}

\begin{lemma}
\label{lem:density-oracle}
There is a polynomial-time $O(\log n)$-approximate density oracle.
\end{lemma}

\begin{proof}
    The proof uses a result by Bateni et al \cite{bateni2018improved}. We provide the full proof in the appendix.
\end{proof}
\subsection{Computing a $\rho$-Approximate Solution of the Assignment Version of the Configuration LP}
\label{sec:fractional}
We next explain how to approximately solve the configuration LP using the
rooted minimum-density oracle from the previous subsection.
We follow the standard oracle-based approach for large covering LPs: we view the configuration LP as a covering problem over an implicit column set and apply the framework of Plotkin--Shmoys--Tardos~\cite{plotkin1991focs};
the key point is that the corresponding pricing subproblem is exactly a rooted minimum-density connected-set problem. Therefore, a $\rho$-approximate density oracle yields, in polynomial time, a $(1 + \varepsilon)\rho$-approximate fractional solution to the configuration LP.
\begin{theorem}
%[Approximate Configuration LP via MWU column generation]
\label{thm:separation}
Assume a $\rho$-approximate density oracle.
Then for any $\varepsilon\in(0,1)$, one can compute in polynomial time a
feasible fractional solution $\hat y$ to the Configuration LP
\eqref{eq:clp-obj}--\eqref{eq:clp-box} satisfying
\[
\sum_{c\in C}\sum_{T\in\mathcal{T}_c} \mathrm{cost}_c(T) \hat y_{T,c} \le (1 + \varepsilon) \rho \cdot \OPT_{\mathrm{CLP}} \le (1 + \varepsilon)\rho\cdot \OPT,
\]
with support size $\mathrm{poly}(n, \varepsilon^{-1}, \rho)$.
\end{theorem}

\begin{proof}
We cast the configuration LP as a pure covering LP and apply a standard
fractional packing/covering approximation scheme (Plotkin--Shmoys--Tardos~\cite{plotkin1991focs}), whose guarantee is that
an approximately optimal \emph{fractional} solution can be computed given
oracle access to an approximate \emph{pricing} subproblem.

We discard all columns $(c,T)$ with $T\cap (V\setminus C) = \emptyset$.
Such columns cover no covering constraint and therefore play no role in the
covering LP. After this restriction, every remaining column satisfies
$\mathrm{cost}_c(T)>0$.

\subparagraph*{Step 1: Configuration LP as a covering LP.}
Let $U_0 \coloneqq V\setminus C$ and $n_0 \coloneqq |U_0|$.
Index columns by pairs $(c,T)$ with $c\in C$ and $T\in\mathcal{T}_c$.
Define the incidence matrix $A\in\{0,1\}^{n_0\times N}$ and cost vector
$c\in\mathbb{R}_{>0}^N$ by
\[
    A_{v,j} \coloneqq \mathbf{1}[v\in T], \qquad c_j \coloneqq \mathrm{cost}_c(T) = \sum_{u\in T}d(u,c).
\]
Since $d$ is a metric with $d(u,c)>0$ for all $u\ne c$, every column
$j=(c,T)$ with $T\ne\{c\}$ satisfies $c_j>0$, so there are no $0/0$
anomalies in the pricing step.
The configuration LP \eqref{eq:clp-obj}--\eqref{eq:clp-box} is then
$\min\{c^\top y : Ay\ge\mathbf{1},y\ge 0\}$, a non-negative covering LP
with $m:=n_0$ constraints and $N$ columns.

\subparagraph*{Step 2: Pricing is (approximate) rooted minimum-density.}
In the Plotkin--Shmoys--Tardos framework for covering LPs, the algorithm maintains a
non-negative weight vector $p\in\mathbb{R}_{\ge 0}^{U_0}$ on constraints and,
in each iteration, needs a column whose \emph{cost per covered weight} is
minimum (equivalently, whose covered weight per unit cost is maximum).
For a column $j=(c,T)$, the covered weight under $p$ is
$\langle p, A_{\cdot,j}\rangle = \sum_{v\in T\cap U_0}p_v$, so the ratio is
\[
    \frac{c_j}{\langle p, A_{\cdot,j}\rangle} = \frac{\sum_{u\in T}d(u,c)}{\sum_{v\in T\cap U_0}p_v}.
\]
This is the rooted minimum-density objective (\Cref{def:rooted-density})
with root $c$, costs $w_c(u):=d(u,c)$, and profits $\alpha_v \coloneqq p_v$
(with $\alpha_c \coloneqq 0$ for centers).
A $\rho$-approximate density oracle therefore provides a $\rho$-approximate
pricing oracle for the covering LP: calling it for each $c\in C$ and
returning the column of globally smallest density gives an
$\rho$-approximate minimizer of the ratio over all columns.

\subparagraph*{Step 3: Width and the upper bound $q$ on $\OPT_\mathrm{CLP}$.} To apply \Cref{thm:pst-covering}, we need (i) an upper bound
$q\ge\OPT_\mathrm{CLP}$ and (ii) the associated width
    \begin{equation}
  \label{eq:width}
      W \ge q\cdot\max_{v\in U_0, j=(c,T)}\frac{A_{v,j}}{c_j}
  = q\cdot\max_{j=(c,T): T\ne\{c\}}\frac{1}{\mathrm{cost}_c(T)}.
\end{equation}
 
\emph{Upper bound $q$.}
Construct the following explicit feasible solution: for each $v\in U_0$,
let $P_v$ be a shortest $G$-path from the nearest center $c_v$ to $v$, which exists by the assumption on connectivity of $G$, and set $y_{V(P_v),c_v}=1$
(with remaining variables $0$).
This is feasible (every $v$ is covered) and has cost
\[
    q \coloneqq \sum_{v\in U_0}\mathrm{cost}_{c_v}(V(P_v)) \le \sum_{v\in U_0}\sum_{u\in V(P_v)}d(u,c_v) \le n_0\cdot n\cdot D_{\max},
\]
where $D_{\max} \coloneqq \max_{u,v\in V}d(u,v)$.
Since the "Width W" step below normalizesd so that $\delta_{min}=1$, then $D_{\max}$ conincides with the aspect ratio $\Delta$, which is $\mathrm{poly}(n)$ by Remark 2. Therefore, $q=\mathrm{poly}(n).$

\emph{Width $W$.}
After normalizing $d$ by its minimum positive value so that
$\delta:=\min_{u\ne v}d(u,v)=1$ (this leaves all approximation ratios
unchanged and ensures $c_j\ge 1$ for all non-trivial columns), the width
in~\eqref{eq:width} becomes
\[
  W = q\cdot\max_j\frac{1}{c_j}
  \le q\cdot\frac{1}{\delta}
  = q
  = \mathrm{poly}(n).
\]

\subparagraph*{Step 4: Apply \Cref{thm:pst-covering}.} With $m=n_0$, the $\rho$-approximate pricing oracle from Step~2, the
upper bound $q\ge\OPT_\mathrm{CLP}$ from Step~3, and the width
$W\le q$ from Step~3, \Cref{thm:pst-covering} yields a feasible
$\hat y$ with
\[
  c^\top\hat y
  \le \rho(1 + \varepsilon + \varepsilon^2) \OPT_\mathrm{CLP}
  \le ( 1 + \varepsilon) \rho \cdot \OPT_\mathrm{CLP}
\]
absorbing $\varepsilon^2$ into $\varepsilon$ by replacing $\varepsilon$ with $\varepsilon/2$,
and support size
\[
  |\mathrm{supp}(\hat y)|
  =\widetilde O\left(\frac{n_0 W \rho}{\varepsilon^3}\right)
  =\widetilde O\left(\frac{n_0 q \rho}{\varepsilon^3}\right)
  =\mathrm{poly}(n).
\]
Since $\OPT_\mathrm{CLP}\le\OPT$ (\Cref{lem:clp-relaxation}), the stated
bound $c^\top\hat y\le(1+\varepsilon)\rho\cdot\OPT$ follows.
 
\subparagraph{Runtime.}
Each iteration of the Plotkin--Shmoys--Tardos framework calls the density oracle $k\coloneqq |C|$
times (once per center) and selects the column of globally smallest
density.
The total number of iterations is $\widetilde O(n_0 W\rho/\varepsilon^3)
= \mathrm{poly}(n)$.
After all iterations, the restricted LP on the generated columns is solved
exactly in polynomial time.
Each oracle call is polynomial by \Cref{lem:density-oracle}, so the entire
algorithm runs in polynomial time.
\end{proof}

%The following is a standard reformulation of the fractional covering framework of Plotkin--Shmoys--Tardos~\cite{plotkin1991focs}; see also the explicit restatement with oracle and support guarantees in~\cite[Theorems~3 and 4]{sharma2020approximation}.

The following theorem is a standard oracle-based reformulation of the
fractional covering framework of Plotkin--Shmoys--Tardos; we cite
Sharma~\cite[Theorems~3 and 4]{sharma2020approximation} for a clean
restatement with approximate pricing and support bounds. For completeness, we provide a self-contained proof of \cref{thm:pst-covering} in the appendix. 

\begin{theorem}[Plotkin--Shmoys--Tardos~\cite{plotkin1991focs}; cf. Sharma~{\cite[Theorems~3 and 4]{sharma2020approximation}}]
\label{thm:pst-covering}
Consider the covering linear program
$
  \min\{c^\top y : Ay\ge\mathbf{1},\;y\ge 0\},
$
where $A\in\mathbb{R}_{\ge 0}^{m\times N}$ and $c\in\mathbb{R}_{>0}^N$.
Let $\OPT$ denote its optimal value, and let $q\ge\OPT$ be a given upper
bound on the optimal objective.
Assume that for every weight vector $w\in\mathbb{R}_{\ge 0}^m$, there is a
polynomial-time $\alpha$-approximate pricing oracle returning a column
$j\in[N]$ with
\[
  \frac{c_j}{\sum_i A_{ij}w_i}
  \le \alpha\cdot\min_{j'\in[N]}\frac{c_{j'}}{\sum_i A_{ij'}w_i}.
\]
Let $W\ge q\max_{i,j}(A_{ij}/c_j)$ be a width bound.
Then for every $\varepsilon\in(0,1]$, there is a polynomial-time algorithm
that computes a feasible $\hat y$ with
$
  c^\top\hat y \le \alpha(1+\varepsilon+\varepsilon^2) \OPT.
$
Moreover, $|\mathrm{supp}(\hat y)|=\widetilde O(mW\alpha/\varepsilon^3)$.
\end{theorem}

\begin{remark}
\label{rem:pst-width}
The approximation ratio $\alpha(1+\varepsilon+\varepsilon^2)$ in
\Cref{thm:pst-covering} depends only on the oracle quality $\alpha$
and the accuracy parameter $\varepsilon$; it is \emph{independent} of the
width $W$.
The width $W$ affects only the support size and iteration count, not the
approximation guarantee.
In our setting $W\le q=\mathrm{poly}(n)$, so both quantities are
polynomial.
\end{remark}

\subsection{Rounding the Configuration LP}
\label{sec:rounding}

\begin{theorem}
\label{thm:main-rounding}
Let $\hat y$ be any feasible solution to \eqref{eq:clp-obj}--\eqref{eq:clp-box}
with polynomial support.
There is a randomized polynomial-time algorithm that, with high probability,
outputs a feasible integral overlap solution $(S_c){c\in C}$ covering all of
$V\setminus C$ and its cost is $O(\log n)\cdot
\sum_{c\in C}\sum_{T\in\mathcal{T}_c}\mathrm{cost}_c(T) \hat y_{T,c}$.
\end{theorem}

\begin{proof}
Fix $\lambda\ge 2$.
% FIX (Issue 4b): Explain why restricting to support pairs is correct.
Independently for each pair $(c,T)$ \emph{in the support of $\hat y$}
(which is of polynomial size), sample $T$ with probability
$
p_{T,c} \coloneqq \min\bigl\{1,\ \lambda\ln(n)\cdot\hat y_{T,c}\bigr\}.
$
For each $c\in C$, let $\mathcal{F}_c$ be the (multi)set of the sampled configurations
and define
$
S_c \coloneqq \{c\}\cup\bigcup_{T\in\mathcal{F}_c} T.
$
Adding $\{c\}$ ensures $c\in S_c$ even if $\mathcal{F}_c = \emptyset$; since
$d(c,c)=0$, this adds zero cost.
Each sampled $T$ is connected and contains $c$, so $S_c$ is a union of
connected sets sharing $c$, hence connected.

\emph{Coverage.}
Fix $v\in V\setminus C$.
Feasibility of $\hat y$ gives $\sum_{c,T:v\in T}\hat y_{T,c}\ge 1$.
If any $p_{T,c}=1$ for a pair with $v\in T$, then $v$ is certainly covered.
Otherwise $p_{T,c}=\lambda\ln(n)\hat y_{T,c}$ for all covering pairs, and
\[
\Pr[v\text{ uncovered}]
= \prod_{c,T:\,v\in T}(1-p_{T,c})
\le \exp\Bigl(-\sum_{c,T:\,v\in T}p_{T,c}\Bigr)
\le \exp(-\lambda\ln n_0)
= n_0^{-\lambda}.
\]
Union bound over $v\in V\setminus C$ (there are $n_0$ such vertices):

$$\Pr[\text{any }v\text{ uncovered}]\le n_0\cdot n_0^{-\lambda}\le n_0^{-(\lambda-1)}\le 1/n_0 = 1/\Omega(n)$$
for $\lambda\ge 2$.

\emph{Cost.}
For each $c$, since the cost of a union is at most the sum of costs:
\[
\sum_{v\in S_c}d(v,c)
\;\le\;\sum_{T\in\mathcal{F}_c}\sum_{v\in T}d(v,c)
= \sum_{T}\mathbf{1}[T\text{ sampled}]\cdot\mathrm{cost}_c(T).
\]
Taking expectation and summing over $c$:
\[
\mathbb{E}[\mathrm{ALG}]
\le \sum_{c,T}\mathrm{cost}_c(T)\cdot p_{T,c}
\le \lambda\ln(n)\sum_{c,T}\mathrm{cost}_c(T) \hat y_{T,c}
= O(\log n)\sum_{c,T}\mathrm{cost}_c(T)\hat y_{T,c}.
\]

By Markov's inequality, a single run has cost at most twice its expectation with constant probability; therefore, by performing $O(\log n)$ independent runs and returning the minimum-cost feasible solution, we obtain an $O(\log n)$-approximation with high probability.
\end{proof}

\subsection{Proof of \cref{thm:assignment-version}}

Now we are ready to complete the proof of \cref{thm:assignment-version}.

\begin{proof}[Proof of \cref{thm:assignment-version}]
    By combining \cref{thm:separation} and \cref{thm:main-rounding}, we obtain that if there exists a $\rho$-approximate density oracle, then there is a polynomial-time algorithm that computes an $O(\rho \log n)$-approximate connected $k$-median clustering. For general graphs, \cref{lem:density-oracle} provides an $O(\log n)$-approximate density oracle. Substituting $\rho = O(\log n)$, we obtain an $O(\log^2 n)$-approximation.
\end{proof}

\iffalse

\begin{corollary}\label{cor:rho-logn}
Assuming a $\rho$-approximate density oracle, there is a randomized polynomial-time
$O(\rho\log n)$-approximation for overlap connected $k$-median with fixed centers.
\end{corollary}

\begin{proof}
Combine \Cref{lem:separation} and \Cref{thm:main-rounding}.
\end{proof}

\begin{corollary}
\label{cor:log2}
There exists an $O(\log^2 n)$-approximation for overlap connected $k$-median with fixed centers.
\end{corollary}

\begin{corollary}
\label{cor:logn}
There is an $O(\log n)$-approximation algorithm for overlap connected $k$-median with fixed centers when the connectivity graph belongs to a minor-closed family, in particular, planar graphs.
\end{corollary}
\fi

% \begin{remark}
% A similar approach extends to the version in which the set of $k$ centers is not fixed. In that setting, it yields an $(O(\log n),\, O(\rho \log n))$-bicriteria approximation for overlap $k$-median, where $\rho$ denotes the approximation guarantee of the rooted minimum-density oracle for the underlying connectivity graph $G$.
% \end{remark}
% \kushagra{I think we need a separate section for bicriteria approximation.}

\section{General Version}
In this section, we establish bicriteria approximation guarantees for the general version of the connected $k$-median clustering problem. In particular, we prove \cref{thm:bicriteria}, which we restate below.

\bicriteria*

To prove \cref{thm:bicriteria}, we follow an approach similar to that used for the assignment version. We first formulate a configuration LP for the general case and compute a near-optimal fractional solution in polynomial time. We then round this solution to obtain the desired bicriteria guarantee. We provide the details in the subsequent sections.

\subsection{General Version of the Configuration LP}

\paragraph*{Primal (Configuration) LP:}
\begin{align}
\min \quad & \sum_{c\in V} \sum_{T \in \mathcal{T}_c} \mathrm{cost}_c(T)\, y_{T,c}
\label{eq:clp-obj-gen}\\
\text{s.t.}\quad
& \sum_{c\in V} \sum_{\substack{T \in \mathcal{T}_c \\ v \in T}} y_{T,c} \ge 1
&& \forall v\in \textcolor{blue}{V},
\label{eq:clp-cover-gen}\\
& \sum_{c \in V}\sum_{T\in\mathcal T_c} y_{T,c} \le k \label{eq:k-centers-gen}\\
& y_{T,c} \ge 0
&& \forall c \in V,\ \forall T \in \mathcal{T}_c.
\label{eq:clp-box-gen}
\end{align}

\subsection{Computing a $\rho$-Approximate Solution of the Configuration LP for the General Variant}
\label{sec:fractional-general}

First, we interpret the center-budget constraint as $\sum_{c \in V} \sum_{T \in \mathcal{T}_c} y_{T,c} \le k$,
which bounds the total configuration mass. This serves as a fractional relaxation of the requirement that at most $k$ centers are opened. In particular, the variables $y_{T,c}$ are now defined for all $c \in V$ and all $T \in \mathcal{T}_c$. Let $\OPT_{\mathrm{GCLP}}$ denote the optimal value of this general configuration LP.

\begin{lemma}
\label{lem:general-clp-solve}
Assume a $\rho$-approximate rooted minimum-density oracle for arbitrary
nonnegative node costs and profits.
Then, for every $\varepsilon\in(0,1)$, one can compute in polynomial time a
polynomial-support fractional solution $\hat y$ satisfying
\begin{align}
& \sum_{c\in V}\sum_{\substack{T\in\mathcal T_c\\ v\in T}}\hat y_{T,c}\ge 1
&& \forall v\in V, \label{eq:gen-frac-cover}\\
& \sum_{c\in V}\sum_{T\in\mathcal T_c}\hat y_{T,c}\le (1+\varepsilon)k,
\label{eq:gen-frac-budget}\\
& \sum_{c\in V}\sum_{T\in\mathcal T_c}\mathrm{cost}_c(T)\hat y_{T,c}
\le (1+\varepsilon)\rho\cdot \OPT_{\mathrm{GCLP}}.
\label{eq:gen-frac-cost}
\end{align}
Moreover, $\hat y$ may be chosen with support size polynomial in the input
size and $1/\varepsilon$.
\end{lemma}

\begin{proof}
Write the general configuration LP as a mixed packing/covering program with
one packing constraint. The covering constraints are indexed by vertices
$v\in V$, and the columns are indexed by pairs $(c,T)$ with $c\in V$ and
$T\in\mathcal T_c$. For such a column, define
\[
A_{v,(c,T)} \coloneqq \mathbf 1[v\in T],
\qquad
b_{(c,T)} \coloneqq 1,
\qquad
q_{(c,T)} \coloneq \mathrm{cost}_c(T).
\]
Thus the LP has the form
\[
\min \sum_{c,T} q_{(c,T)}y_{T,c}
\quad\text{s.t.}\quad
Ay\ge \mathbf 1,\qquad
\sum_{c,T} b_{(c,T)}y_{T,c}\le k,\qquad
y\ge 0.
\]

We apply the standard oracle-based framework for mixed packing and covering
programs, as in Young~\cite{young2001sequential} and
Plotkin--Shmoys--Tardos~\cite{plotkin1991focs}. It suffices to identify the
pricing problem required by this framework.

For a guessed objective value $\tau$, the feasibility version is
\[
Ay\ge \mathbf 1,\qquad
\sum_{c,T} y_{T,c}\le k,\qquad
\sum_{c,T}\mathrm{cost}_c(T)y_{T,c}\le \tau,\qquad
y\ge 0.
\]
In the mixed packing/covering framework, a pricing step is given nonnegative
weights $p_v$ on the covering constraints and nonnegative multipliers
$\eta,\theta$ for the two packing constraints. It asks for a column minimizing
the ratio
\[
\frac{\eta\cdot b_{(c,T)}+\theta\cdot q_{(c,T)}}
     {\sum_{v\in V}p_v A_{v,(c,T)}} = \frac{\eta+\theta\,\mathrm{cost}_c(T)}
     {\sum_{v\in T}p_v}.
\]
For a fixed root $c$, this is exactly a rooted minimum-density connected-set
problem. Indeed, define node costs $w_c^{\eta,\theta}(u):=\theta\,d(u,c)+\eta\cdot\mathbf 1[u=c]$, and node profits $\alpha_u \coloneqq p_u$. Since every $T\in\mathcal T_c$ contains
$c$, we have
\[
\sum_{u\in T} w_c^{\eta,\theta}(u) = \theta\sum_{u\in T}d(u,c) + \eta = \theta \mathrm{cost}_c(T)+\eta.
\]
Therefore the pricing ratio is precisely
$
\frac{\sum_{u\in T} w_c^{\eta,\theta}(u)}
     {\sum_{v\in T}\alpha_v}.
$
Calling the rooted minimum-density oracle once for each possible root
$c\in V$ and taking the best returned column gives a $\rho$-approximate
pricing oracle for the mixed packing/covering instance.

By the standard guarantee for mixed packing and covering with a
$\rho$-approximate pricing oracle, binary search over $\tau$ and the above
pricing routine compute, in polynomial time, a solution satisfying
\eqref{eq:gen-frac-cover}, \eqref{eq:gen-frac-budget}, and
\eqref{eq:gen-frac-cost}. The factor $(1+\varepsilon)$ accounts for the
usual accuracy loss in the mixed packing/covering feasibility procedure.

Finally, the algorithm generates only polynomially many columns. Solving the
restricted mixed LP over the generated columns and taking a basic feasible
solution yields a polynomial-support solution $\hat y$ with the same bounds.
\end{proof}

\begin{corollary}
\label{cor:general-frac-logn}
For general connectivity graphs, one can compute in polynomial time a
fractional solution $\hat y$ satisfying \eqref{eq:gen-frac-cover},
\eqref{eq:gen-frac-budget}, and
\[
\sum_{c\in V}\sum_{T\in\mathcal T_c}\mathrm{cost}_c(T)\hat y_{T,c}
\le O(\log n)\cdot \OPT_{\mathrm{GCLP}}.
\]
\end{corollary}

\begin{proof}
By \cref{lem:density-oracle}, the rooted minimum-density oracle has
approximation ratio $\rho=O(\log n)$ on general graphs. Applying
\cref{lem:general-clp-solve} with a constant $\varepsilon$ gives the claim.
\end{proof}

\subsection{Rounding the General Version of the Configuration LP}
\label{sec:rounding-general}

In this subsection, we round the fractional solution obtained from \cref{lem:general-clp-solve}.

\begin{lemma}\label{lem:bicriteria-nonfixed}
Suppose we are given a polynomial-support feasible solution $\hat{y}$ to
\eqref{eq:clp-obj-gen}--\eqref{eq:clp-box-gen} of cost at most $\Gamma$.
Then there is a randomized polynomial-time algorithm that outputs, with high probability,
an overlapping connected clustering using at most $O(k \log n)$ centers of cost 
$ O(\log n) \cdot \Gamma$.
\end{lemma}

\begin{proof}
Fix a constant $\lambda \geq 2$. Independently for every pair $(c,T)$ in the support of
$\hat{y}$, sample the configuration $T$ with probability $p_{T,c} \coloneqq \min\{1, \lambda \ln n \cdot \hat{y}_{T,c}\}$. For each $c \in V$, let $\mathcal{F}_c$ be the collection of the sampled configurations rooted at
$c$. We open center $c$ if $\mathcal{F}_c \neq \emptyset$, and define its cluster as $S_c \coloneqq \bigcup_{T \in \mathcal{F}_c} T$. Since every sampled set $T \in \mathcal{T}_c$ contains $c$ and is connected, the union
$S_c$ is connected whenever $\mathcal{F}_c \neq \emptyset$.

We first prove coverage. Fix a vertex $v \in V$. By feasibility of $\hat{y}$,
$
    \sum_{c \in V} \sum_{\substack{T \in \mathcal{T}_c \\ v \in T}} \hat{y}_{T,c} \geq 1.
$
If there exists a pair $(c,T)$ with $v \in T$ and $p_{T,c}=1$, then $v$ is covered with
probability one. Otherwise,
\begin{align*}
    \Pr[v \text{ is uncovered}]
    = \prod_{c,T: v \in T} (1-p_{T,c})
    &\leq \exp\Big(-\sum_{c,T: v \in T} p_{T,c}\Big) \\
    &= \exp\Big(-\lambda \ln n 
        \sum_{c,T: v \in T} \hat{y}_{T,c}\Big) \leq n^{-\lambda}.
\end{align*}
Taking a union bound over all vertices, the probability that some vertex is uncovered is at
most $n^{1-\lambda}$, which is at most $1/n$ for $\lambda \geq 2$.

Next, we bound the number of opened centers. Let $N$ be the number of the sampled
configurations. Since the number of opened centers is at most $N$, it suffices to bound $N$.
We have
\[
    \mathbb{E}[N]
    = \sum_{c \in V} \sum_{T \in \mathcal{T}_c} p_{T,c}
    \leq \lambda \ln n \sum_{c \in V} \sum_{T \in \mathcal{T}_c} \hat{y}_{T,c}
    \leq \lambda (1 + \varepsilon) k \ln n,
\]
where the last inequality follows from Constraint~\eqref{eq:gen-frac-budget}.
By a Chernoff bound, $N = O(k \log n)$ with high probability. Therefore, the algorithm
opens at most $O(k \log n)$ centers with high probability.

Finally, we bound the cost. Since the cost of a union is at most the sum of the costs of
the sampled configurations,
\[
    \cost(\{S_c\})
    = \sum_{c:\mathcal{F}_c \neq \emptyset} \sum_{v \in S_c} d(v,c) \leq \sum_{c \in V} \sum_{T \in \mathcal{F}_c} \cost_c(T).
\]
Taking expectation,
\[
    \mathbb{E}[\cost(\{S_c\})]
    \leq \sum_{c \in V} \sum_{T \in \mathcal{T}_c} p_{T,c}\,\cost_c(T) \leq \lambda \ln n         \sum_{c \in V} \sum_{T \in \mathcal{T}_c}         \cost_c(T)\hat{y}_{T,c} = O(\log n)\cdot \Gamma.
\]
By Markov's inequality, a single run has cost at most twice its expectation with constant probability; therefore, by performing $O(\log n)$ independent runs and returning the minimum-cost feasible solution, we obtain an $O(\log n)$-approximation with high probability. This proves the claimed bicriteria guarantee.
\end{proof}

\subsection{Proof of \cref{thm:bicriteria}}

Now, we are ready to complete the proof of \cref{thm:bicriteria}.

\begin{proof}[Proof of \cref{thm:bicriteria}]
Let $\hat y$ be the fractional solution returned by
\cref{lem:general-clp-solve} with a fixed constant $\varepsilon\in(0,1)$.
Then $\hat y$ satisfies the covering constraints and
$
\sum_{c\in V}\sum_{T\in\mathcal T_c}\hat y_{T,c}\le (1+\varepsilon)k,
$
and its cost is at most
$
(1+\varepsilon)\rho\cdot \OPT_{\mathrm{GCLP}}
\le
(1+\varepsilon)\rho\cdot \OPT,
$
where the last inequality follows because the general configuration LP is a
relaxation of the non-assignment overlapping connected $k$-median problem.

Now apply \cref{lem:bicriteria-nonfixed} to $\hat y$ with
$
\Gamma := (1+\varepsilon)\rho\cdot \OPT.
$
The rounding opens at most
$
O\bigl((1+\varepsilon)k\log n\bigr)=O(k\log n)
$
centers with high probability, and its expected cost is at most
$
O(\log n)\cdot \Gamma
=
O(\rho\log n)\cdot \OPT.
$
Thus we obtain an $\bigl(O(\log n),\,O(\rho\log n)\bigr)$ bicriteria approximation.

For general graphs, \cref{lem:density-oracle} gives $\rho=O(\log n)$.
Therefore, the bicriteria guarantee becomes
$
\bigl(O(\log n),\,O(\log^2 n)\bigr).
$
This proves \cref{thm:bicriteria}.
\end{proof}

\section{Conclusion and Open Problems}

In this paper, we developed a configuration-LP framework for the connected $k$-median problem. 
Unlike previous cut-based relaxations, our formulation incorporates connectivity directly into the 
LP variables by working with connected configurations rooted at centers. We combined this viewpoint 
with covering-LP techniques and rooted minimum-density oracles to obtain approximation guarantees 
for both the assignment version and a bicriteria variant of the general version. In particular, we 
obtained an $O(\log^2 n)$-approximation for the assignment version on general graphs, together with 
an $(O(\log n), O(\log^2 n))$ bicriteria approximation for the general version.

More broadly, our results suggest that configuration LPs provide a natural framework for handling 
connectivity constraints in clustering problems. By separating the problem into a covering-LP layer 
and a graph-theoretic rooted density oracle, the framework makes it possible to incorporate stronger 
oracles for special graph families and related connectivity models. We believe this viewpoint may be 
useful beyond connected $k$-median, particularly for other constrained clustering objectives where 
connectivity plays a central role.

Several interesting open problems remain. The most natural question is whether one can obtain polylogarithmic approximation 
guarantees for the general version containing exactly $k$ centers. More generally, it would be interesting to better understand the power and 
limitations of configuration-LP formulations for connectivity-constrained clustering problems. 
Another promising direction is to design stronger rooted minimum-density oracles for restricted graph 
families, such as minor-closed graphs. Since the rooted minimum-density problem already admits an 
$O(\log n)$-approximation in general graphs, improving this guarantee on structured graph families 
would immediately translate into stronger approximation guarantees within our framework.

\newpage

\printbibliography
\appendix

\section{Proof of \cref{lem:clp-relaxation}}

\begin{proof}
Fix a feasible overlapping solution $(S_c)_{c\in C}$.
By feasibility, for every $c\in C$, we have $c\in S_c$ and $G[S_c]$ is connected.
Hence $S_c\in \mathcal{T}_c$ for every center $c$.
Define a configuration-LP solution $y$ by
\[
y_{T,c}:=
\begin{cases}
1, & \text{if } T=S_c,\\
0, & \text{otherwise}.
\end{cases}
\]
We verify feasibility.
First, for every $v\in V\setminus C$, the overlapping solution satisfies
$\bigcup_{c\in C} S_c = V$, so there exists at least one center $c$ such that
$v\in S_c$. Therefore,
$
\sum_{c\in C}\sum_{\substack{T\in\mathcal{T}_c\\ v\in T}} y_{T,c}
\;\ge\;
1,
$
and thus the covering constraints \eqref{eq:clp-cover} hold for all
$v\in V\setminus C$.
Second, by construction, all variables satisfy $y_{T,c}\ge 0$, so
\eqref{eq:clp-box} also holds.
Finally, the objective value of this LP solution is
\[
\sum_{c\in C}\sum_{T\in\mathcal{T}_c} \mathrm{cost}_c(T)\,y_{T,c}
=
\sum_{c\in C} \mathrm{cost}_c(S_c)
=
\sum_{c\in C}\sum_{v\in S_c} d(v,c),
\]
which is exactly the cost of the given integral solution.

\smallskip
Since every feasible solution to the assignment variant of connected $k$-median with overlapping clusters induces a feasible solution to the configuration LP with the same cost, the optimal value of the configuration LP is at most $\OPT$, the optimal cost of connected $k$-median with overlapping clusters for the prescribed center set $C$.
\end{proof}

\section{Proof of \cref{lem:density-oracle}}

\begin{proof}
Write $w(T):=\sum_{v\in T}w(v)$ and $\alpha(T):=\sum_{v\in T}\alpha_v$.
Let $T^*$ be an optimal density solution with $\lambda^*:=w(T^*)/\alpha(T^*)$.
 
We use the bicriteria guarantee of Bateni et al.~\cite[Theorem~3]{bateni2018improved} for the {\em rooted budgeted node-weighted Steiner tree problem}. In this problem, we are given a graph $G=(V,E)$, a designated root $c\in V$, a non-negative node-cost function $w:V\to\mathbb{R}_{\ge 0}$, a non-negative node-profit function $\alpha:V\to\mathbb{R}_{\ge 0}$, and a budget
$B\ge 0$. The goal is to find a connected tree $T\subseteq G$ containing $c$ whose total node cost is at most $B$ and whose total profit $\alpha(T):=\sum_{v\in T}\alpha(v)$
is as large as possible. Let
\[
\alpha^*(B):=\max\{\alpha(T): T\ni c \text{ is connected and } w(T)\le B\}
\text{ where } w(T):=\sum_{v\in T} w(v).
\]
Then, for any budget $B$ and any $\varepsilon>0$, there is a polynomial-time
algorithm that returns a connected tree $T_B\ni c$ such that
$
w(T_B)\le (1+\varepsilon)B$ and $
\alpha(T_B)\ge \frac{\varepsilon}{O(\log n)}\cdot \alpha^*(B).
$
In other words, the algorithm may violate the budget by a factor of at most
$1+\varepsilon$, while guaranteeing an $O(\log n/\varepsilon)$-approximation
to the maximum achievable profit under budget $B$.
 
\medskip
Set $W_{\max}:=w(V)$ and try the $O(\log W_{\max})$ budgets $B\in\{1,2,4,\ldots,2^{\lceil\log_2 W_{\max}\rceil}\}$. 
There exists $B^\star$ in this list with $w(T^*)\le B^\star\le 2w(T^*)$.
Apply the guarantee of Bateni et al.~\cite{bateni2018improved} with $\varepsilon:=1$.
Since $w(T^*)\le B^\star$, we have $\alpha^*(B^\star)\ge\alpha(T^*)$.
The algorithm returns $T_{B^\star}\ni c$ with
$
w(T_{B^\star})\le 2B^\star\le 4w(T^*)$ and $
\alpha(T_{B^\star})\ge \alpha(T^*)/O(\log n).
$
Therefore,
\[
  \dens(T_{B^\star})
  = \frac{w(T_{B^\star})}{\alpha(T_{B^\star})}
  \le \frac{4\,w(T^*)}{\alpha(T^*)/O(\log n)}
  = O(\log n)\cdot\lambda^*.
\]
Returning $\arg\min_B \dens(T_B)$ over all tried budgets yields an
$O(\log n)$-approximate density oracle in polynomial time.
\end{proof}

\section{Proof of \cref{thm:pst-covering}}
% ======================================================================
%  A more explanatory proof of Theorem 9 (same statement, same proof
%  strategy as the original Appendix C proof — no change in content,
%  only in exposition). Drop-in replacement for the current proof.
%  Requires the same preamble macros as before.
% ======================================================================

\begin{proof}
 The covering LP $\min\{c^\top y : Ay\ge \mathbf 1,\ y\ge0\}$
has (in our usage) exponentially many columns, so we cannot solve it
directly; we only have oracle access to an approximately-best column for
any given constraint weighting. The proof has three parts.

\begin{enumerate}
\item We first observe that the $\alpha$-approximate \emph{pricing}
oracle we are given is indeed also an $\eta$-weak \emph{index-finding} oracle
($\eta=1/\alpha$) for a natural quantity $D_j(w)$, the ``coverage per unit
weight'' of column $j$.
\item We then show that an index-finding oracle over columns is the
same thing as a point-finding oracle over a certain simplex-like polytope
$P_r$ (all vectors of cost exactly $r$). This is the step that lets us
invoke a known black-box algorithm for fractional covering, which expects
a point-finding oracle, not a column oracle.
\item The black-box algorithm, for a \emph{fixed} guess $r$ of the
optimal cost, either certifies cost $r$ suffices (up to a small violation) or fails. We do not know $\mathrm{OPT}$ in advance, so we do a 
binary search over $r$; two monotonicity facts about this success/failure
signal pin down $r\approx\mathrm{OPT}$ and let us convert the near-feasible
solution found at the end into a genuinely feasible one, at a controlled
extra cost.
\end{enumerate}
In what follows, we elaborate on each of these steps.

\subparagraph*{Part 1: from pricing to index-finding.}
For a column $j$ and weight vector $w\in\mathbb R^m_{\ge0}$, let
$
D_j(w) \ :=\ \frac{w^\top A_{\cdot j}}{c_j} \ =\ \frac{\sum_i A_{ij}w_i}{c_j}.
$
Think of $D_j(w)$ as how much weighted coverage column $j$ buys you per unit of its cost. An oracle is an \emph{$\eta$-weak index-finding
oracle} if, given $w\ge0$, it returns some $j$ with
$D_j(w) \ge \eta\max_{\ell} D_\ell(w)$, in other words a column that is within a
factor $\eta\le1$ of the best possible coverage-per-cost.

The assumed pricing oracle minimizes the \emph{reciprocal} quantity
$c_j/(w^\top A_{\cdot j}) = 1/D_j(w)$ to within a factor $\alpha$:
If 
$
\frac{c_j}{w^\top A_{\cdot j}} \ \le\ \alpha\cdot \min_{j'} \frac{c_{j'}}{w^\top A_{\cdot j'}},
$
then taking reciprocals (all quantities are positive since $c>0$; the
degenerate case $\max_\ell D_\ell(w)=0$, i.e.\ when $w$ is orthogonal to every
column, makes the index-finding requirement vacuous and is set aside), then
\[
D_j(w) \ \ge\ \frac{1}{\alpha}\max_{\ell} D_\ell(w).
\]
So the pricing oracle \emph{is} an $\eta$-weak index-finding oracle with
$\eta \coloneqq 1/\alpha$.

\subparagraph*{Part 2: from index-finding to point-finding on $P_r$.}
Fix a target cost $r\ge0$ and consider the cost-$r$ slice of the nonnegative orthant,
$
P_r \ :=\ \{y\in\mathbb R^N_{\ge0} : c^\top y = r\}.
$
We introduce this set because known fractional-covering algorithms are
stated abstractly for a polytope $P$ together with a \emph{point-finding}
oracle (return an approximately-best point of $P$ for a given linear
objective), and we want to feed our column oracle into one of them. The
set $P_r$ is exactly the right polytope: every $y\in P_r$ can be written as
$
y \ = \ \sum_{j=1}^N \frac{c_jy_j}{r}\cdot\frac{r}{c_j}e_j,
$
and the coefficients $c_jy_j/r$ are nonnegative and sum to $1$ (since
$\sum_j c_jy_j = c^\top y = r$), so $y$ is a convex combination of the
points $\{(r/c_j)e_j : j\in[N]\}$. In other words, $P_r$ is the
simplex whose vertices are, for each column $j$, spend the entire budget $r$ on column $j$ alone. Maximizing a linear function over a simplex
is achieved at a vertex, so for any weight vector $w\ge0$,
\[
\max_{y\in P_r} w^\top Ay = \max_{j\in[N]} w^\top A\Big(\frac r{c_j}e_j\Big) = r\cdot\max_{j} D_j(w).
\]
This identity is the crux of this part: finding a good \emph{point} of $P_r$ for objective $w^\top A(\cdot)$ is literally the same problem as
finding a good \emph{column} for $D_j(w)$, just rescaled by $r$. So our
$\eta$-weak index-finding oracle, applied to $w$, returns a column $j$ with
$D_j(w)\ge \eta\max_\ell D_\ell(w)$; setting $y := (r/c_j)e_j \in P_r$ gives
\[
w^\top Ay = r D_j(w) \ \ge\ \eta \cdot r\max_\ell D_\ell(w) = \eta\max_{y'\in P_r} w^\top Ay',
\]
In other words, $y$ is an $\eta$-weak maximizer of $w^\top Ay$ over $P_r$. We have
therefore converted our column oracle into an $\eta$-weak
point-finding oracle for $P_r$, for every $r$ simultaneously (the
construction above does not depend on $r$ except through the final
rescaling).

We also record the \emph{width} of this setup, since the black-box
algorithm's running time depends on it: by the same vertex computation,
\[
\mathrm{width}(A,\mathbf 1, P_r) \coloneqq \max_{y\in P_r}\max_i (Ay)_i \ =\ \max_{i,j} A_{ij}\cdot\frac r{c_j} \ =\ r\max_{i,j}\frac{A_{ij}}{c_j}.
\]
Using any known bound $q\ge\mathrm{OPT}$
and the given $W \ge q\max_{i,j}A_{ij}/c_j$, we get, for every $r\le q$,
\[
\mathrm{width}(A,\mathbf 1,P_r) = r\max_{i,j}\frac{A_{ij}}{c_j} \le q\max_{i,j}\frac{A_{ij}}{c_j} \le W.
\]

\medskip
\noindent\textbf{Part 3: black-box covering, then binary search on $r$.}
We now invoke, as a black box, the standard guarantee for fractional
covering with an approximate point-finding oracle (Plotkin--Shmoys--Tardos;
we use the clean restatement of Sharma \cite{sharma2020approximation}, Theorems 3--4):
given the $\eta$-weak point-finding oracle for $P_r$ from Part 2 and the
width bound $W$ from Part 2, and setting
$
\mu \coloneqq \frac{\eta}{1+\epsilon},
$
the algorithm either correctly reports $\{y\in P_r : Ay\ge\mathbf 1\}$ is
empty or returns some $y\in P_r$ with $Ay \ge \mu\mathbf 1$ (i.e. every
covering constraint is satisfied up to a $(1-\mu)$ shortfall) using a
number of oracle calls that yields support size $|\mathrm{supp}(y)| =
\widetilde O(mW/(\eta\epsilon^3))$ when the oracle always returns a
single vertex (which ours does, by construction in Part 2). Call this
routine on input $r$ and write $g(r)=1$ if it succeeds (returns such a
$y$) and $g(r)=0$ if it reports infeasibility.

The routine above only tells us, for a \emph{given} $r$, whether spending exactly $r$ can nearly cover everything.
We do not know $\mathrm{OPT}$, so we search for the smallest $r$ at which $g(r)$ flips from $0$ to $1$. The bisection argument has two parts:

\emph{(i) $r\ge \mathrm{OPT} \implies g(r)=1$.} Let $y^\ast$ be an optimal
solution of the original covering LP, so $c^\top y^\ast = \mathrm{OPT}$ and
$Ay^\ast\ge\mathbf1$. Rescaling, $\tfrac{r}{\mathrm{OPT}}y^\ast \in P_r$
(it has cost exactly $r$) and $A\big(\tfrac{r}{\mathrm{OPT}}y^\ast\big) =
\tfrac{r}{\mathrm{OPT}}Ay^\ast \ge \tfrac{r}{\mathrm{OPT}}\mathbf 1 \ge
\mathbf 1$ since $r\ge\mathrm{OPT}$. So $\{y\in P_r : Ay\ge \mathbf1\}$ is
genuinely nonempty, and since it is nonempty (and in particular $Ay\ge\mathbf1\ge\mu\mathbf1$ trivially, as $\mu\le1$), the routine cannot
correctly report infeasibility; hence $g(r)=1$.

\emph{(ii) $r < \mu\cdot\mathrm{OPT} \implies g(r)=0$.} Suppose instead
$g(r)=1$, i.e.\ the routine returns $y\in P_r$ with $Ay\ge\mu\mathbf1$.
Rescale: $\hat y := y/\mu$ satisfies $A\hat y = Ay/\mu \ge \mathbf1$, so
$\hat y$ is feasible for the \emph{original} covering LP, with cost
$c^\top\hat y = (c^\top y)/\mu = r/\mu < \mathrm{OPT}$ (using $r <
\mu\cdot\mathrm{OPT}$). This contradicts optimality of $\mathrm{OPT}$.
So $g(r)=1$ is impossible, i.e.\ $g(r)=0$.

Then, we initialize
$\ell=0$, $u=q$ (both satisfy $g(\ell)=0,g(u)=1$, using fact (ii) at
$\ell=0$ whenever $\mathrm{OPT}>0$), and repeatedly query the midpoint
$r=(\ell+u)/2$: if $g(r)=0$ set $\ell\coloneqq r$, otherwise set $u\coloneqq r$ and record
the returned vector. Stop once $u\le(1+\delta)\ell$ for
$\delta \coloneqq \varepsilon^2/(1+\varepsilon)$; this takes
$O(\log(q/(\delta\cdot(\text{smallest gap}))))$ iterations, polynomial in
the input size.

Finally, at termination $g(\ell)=0$ forces
$\ell<\mathrm{OPT}$ (else fact (i) would give $g(\ell)=1$), so
$u \le (1+\delta)\ell < (1+\delta)\,\mathrm{OPT}$.
Let $y^{(u)}\in P_u$ be the vector recorded at the last successful
call ($g(u)=1$), so $Ay^{(u)}\ge\mu\mathbf1$ and $c^\top y^{(u)} = u$.
Define $\hat y \coloneqq \frac{1}{\mu}y^{(u)}$.
Exactly as in fact (ii)'s rescaling, $A\hat y \ge \mathbf 1$, so $\hat y$ is
feasible for the original covering LP. Its cost is
$c^\top\hat y = \frac{u}{\mu} \le \frac{(1+\delta)}{\mu} \mathrm{OPT}$.
Substituting $\mu = \eta/(1 + \varepsilon)$ and $\delta = \varepsilon^2/(1+\varepsilon)$:
\[
\frac{1+\delta}{\mu} \ =\ \Big(1+\frac{\varepsilon^2}{1 + \varepsilon}\Big)\cdot\frac{1 + \varepsilon}{\eta} = \frac{1 + \varepsilon + \varepsilon^2}{\eta} = \alpha(1 + \varepsilon + \varepsilon^2),
\]
using $\eta=1/\alpha$ from Part 1. This gives the claimed cost bound
$c^\top\hat y \le \alpha(1 + \varepsilon + \varepsilon^2) \mathrm{OPT}$. Rescaling by
$1/\mu$ does not change the support, so
\[
|\supp(\hat y)| = |\supp(y^{(u)})| = \widetilde O\Big(\frac{mW}{\eta\varepsilon^3}\Big) = \widetilde O\Big(\frac{mW\alpha}{\varepsilon^3}\Big),
\]
matching the statement of the theorem.

\subparagraph*{Runtime.} Each bisection step invokes the black-box covering routine once, which in turn calls our point-finding oracle
(equivalently, the original pricing oracle, once per iteration of that
routine) $\mathrm{poly}(m,W,1/\varepsilon)$ times; the number of bisection
steps is $\mathrm{poly}(\log(q/\mathrm{OPT}), \log(1/\delta))$. All oracle
calls and arithmetic operations are polynomial whenever $q$, $W$, and the
input encoding lengths are polynomially bounded, giving the stated overall
polynomial running time.
\end{proof}

\end{document}